\documentclass[letterpaper, 10pt, conference]{ieeeconf}
\usepackage{graphicx} % Required for inserting images

\IEEEoverridecommandlockouts

\def\usingarxiv{false}

\usepackage{multibib}
\newcites{app}{Appendix References}

\usepackage{amsmath}
\usepackage{amssymb}
\usepackage{amsthm}
\usepackage{algorithm}
\usepackage{algpseudocode}
\usepackage{etoolbox} % for the arxiv macro
\usepackage[hidelinks]{hyperref} % for urls

\newcommand{\arxiv}[2]{\ifbool{\usingarxiv}{#1}{#2}} % checks if we are compiling for arxiv or not

\newcommand{\norm}[1]{\rVert #1\lVert}
\newcommand{\innerp}[1]{\langle #1 \rangle}
{
    \theoremstyle{plain}
    \newtheorem{assumption}{Assumption}
    \newtheorem{proposition}{Proposition}
    \newtheorem{theorem}{Theorem}
    \newtheorem{lemma}{Lemma}
    \newtheorem*{remark}{Remark}
    
    \newtheorem{definition}{Definition}
}

\DeclareMathOperator*{\diag}{diag}

\makeatletter
\newcommand\fs@betterruled{%
  \def\@fs@cfont{\bfseries}\let\@fs@capt\floatc@ruled
  \def\@fs@pre{\vspace*{7pt}\hrule height.8pt depth0pt \kern2pt}%
  \def\@fs@post{\kern2pt\hrule\relax}%
  \def\@fs@mid{\kern2pt\hrule\kern2pt}%
  \let\@fs@iftopcapt\iftrue}
\floatstyle{betterruled}
\restylefloat{algorithm}
\makeatother

\title{\LARGE \bf Bilevel Optimization for Real-Time Control \\ with Application to Locomotion Gait Generation}
\author{
    Zachary Olkin$^{1}$ and Aaron D. Ames$^{1}$%
    \thanks{This research is supported by Technology Innovation Institute (TII).}
    \thanks{$^{1}$Authors are with the Department of Control and Dynamical Systems, California Institute of Technology, Pasadena CA 91125, U.S.A. \texttt{\{zolkin, ames\}@caltech.edu}.}
}
\date{February 2024}

\begin{document}
 
\maketitle

\thispagestyle{empty}
\pagestyle{empty}

\renewcommand\qedsymbol{$\blacksquare$}

\begin{abstract}
    Model Predictive Control (MPC) is a common tool for the control of nonlinear, real-world systems, such as legged robots. However, solving MPC quickly enough to enable its use in real-time is often challenging.  
    One common solution is given by real-time iterations, which does not solve the MPC problem to convergence, 
    but rather close enough to give an approximate solution. In this paper, we extend this idea to a bilevel control framework where a ``high-level" optimization program modifies a controller parameter of a ``low-level" MPC problem which generates the control inputs and desired state trajectory. We propose an algorithm to iterate on this bilevel program in real-time and provide conditions for its convergence and improvements in stability. We then demonstrate the efficacy of this algorithm by applying it to a quadrupedal robot where the high-level problem optimizes a contact schedule in real-time. We show through simulation that the algorithm can yield improvements in disturbance rejection and optimality, while creating qualitatively new gaits. 
\end{abstract}

\section{Introduction}
Model Predictive Control (MPC) has become an increasingly common method for control as computation has rapidly improved over the past few decades. MPC is a control strategy where an optimization problem is solved online that yields a control input and state trajectory which minimize a given cost function over a fixed finite time horizon. In the MPC problem there may exist control or constraint parameters that can be changed or tuned over time. In this paper, we consider the problem of updating these parameters in real-time through a bilevel optimization framework.

Solving MPC to convergence in real time is difficult for many systems. In response, real-time iterations \cite{diehl_efficient_2002, diehl_real-time_2005} only approximately solve the MPC program with Quadratic Program (QP) approximations at each time step, warm-starting the next solve with the prior solution. In practice, this allows solutions to converge across iterations while retaining real-time performance. Later work has even extended this by only partially solving the above QP approximation with the aim of maximizing the feedback rate at the expense of convergence at each iteration \cite{wang_fast_2010, feller_stabilizing_2016}. For dynamic applications like legged locomotion, these methods are common and performant on real systems \cite{sleiman_unified_2021, grandia_perceptive_2022, csomay-shanklin_nonlinear_2023}.

% START OF RELATED WORK %
Methods for solving generic bilevel optimization programs include transforming the bilevel program into a single level, penalty function methods, and descent methods \cite{sinha_review_2020}. Single level and penalty function approaches cannot be used here as we require an (approximate) solution to the low-level program at each time instant. This paper considers a novel real-time approach using the ideas of descent methods.

A use case for the presented bilevel algorithm occurs in the control of legged robots, which will be the focus of a simulation example later in the paper. In legged locomotion, the robot must use a contact schedule, which we define as the collection of contact timings and contact sequences of the legs, to move about the environment. Determining this contact schedule is challenging in general due to the nonlinear and hybrid nature of the system. In this work, we utilize the observation that every foot is in a binary state with a given surface: in contact or out of contact. Therefore, we can parameterize the contact schedule using the contact times of each foot. Rather than solving for these timings in the MPC problem, we can instead treat them as a parameter to the MPC problem and update them separately in a high-level optimization. This allows us to preserve the speed and reliability of the underlying MPC problem while still modifying the gait in real-time.

\begin{figure} [t]
    \centering
    \includegraphics[scale=1.02]{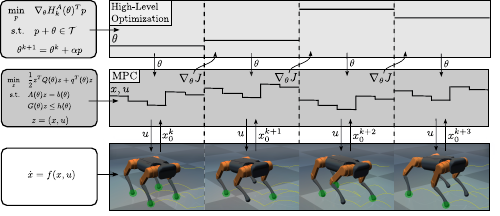}
    \caption{
        Structure of the bilevel optimization. The MPC uses parameters from the high-level optimization and outputs the control inputs and state trajectory. When applied to the quadruped, the high-level parameter is the contact schedule. The green dots indicate contact with the ground and show how the contact schedule changes over time.
    }
    \label{fig:big-idea}
    \vspace{-0.5cm}
\end{figure}

Other approaches to the online contact schedule generation problem include ``Contact-Implicit" MPC (CIMPC) \cite{kong_hybrid_2022, kurtz_inverse_2023, cleach_fast_2023}. In CIMPC, the contact schedule is determined implicitly either through choice of model/contact (i.e. smoothed contact dynamics), or algorithm (i.e. hybrid iLQR). Smoothing the contact dynamics leads to a trade off between the numerical conditioning of the problem and physical realism, \cite{wensing_optimization-based_2022} which this paper's method circumvents. Some CIMPC approaches \cite{carius_trajectory_2018, cleach_fast_2023} solve bilevel optimization problems where the contact is resolved via another optimization problem, which is used as part of the dynamics in a normal MPC formulation. This is structurally different from the algorithm presented here, since we optimize the contact schedule independently from the choice of contact dynamics.

We propose a real-time iteration scheme for a bilevel optimization problem involving a low-level MPC controller as depicted in Fig. \ref{fig:big-idea}. Theoretical properties of the algorithm such as convergence and stability improvements are proven. Further, we illustrate how the algorithm can be used as a novel approach to the online gait optimization problem for legged robots. Finally, we provide simulation results that demonstrate improvements in disturbance rejection and optimality as well as the generation of qualitatively different gaits while walking. We achieve MPC loop rates of around 90 Hz for a quadruped, which is competitive with current state of the art contact schedule generation approaches. 

\section{Algorithm and Theoretical Properties} \label{sec:theory}
In this section, we show that the proposed bilevel algorithm is theoretically justified under nominal conditions.

\subsection{Problem Setup}
Let $N$ denote the number of time steps (nodes) in the time horizon, $x \in \mathbb{R}^{n\times N}$ the state trajectory, $u \in \mathbb{R}^{m\times N}$ the input trajectory, $f$ the discrete-time dynamics, $\mathcal{C}$ the constraint set, and $L$ the low-level cost function. Throughout, we use subscripts to denote a node on a trajectory, e.g., $x_i$ and $u_i$ and a superscript to denote an iterate in time, e.g., $\theta^k$. Consider the parameterized low-level MPC program
\begin{align}
\label{eq:ll_opt}
    \min_{x, u}& \quad L(x, u, \theta) \notag \\
    \text{s.t.}& \quad (x, u) \in \mathcal{C}(\theta), \tag{L-MPC} \\
    & \quad  x_{i+1} = f(x_i, u_i, \theta), \; \forall k < N \notag,
\end{align}
where $\theta \in \mathbb{R}^c$ are the MPC program parameters. The high-level problem is given by
\begin{align}
    \label{eq:hl_opt}
    \min_{\theta}& \quad H(x^*,u^*,\theta) \notag \\
    \text{s.t.} &\quad  \theta \in \mathcal{T}, \\
    &\quad (x^*, \; u^*) \textrm{ is optimal for } \eqref{eq:ll_opt}, \notag
\end{align}
where $H$ gives the high-level cost and $\mathcal{T}$ is the constraint set for the parameter.
As both problems are in general non-convex and potentially large, solving this bilevel problem to convergence in real-time for dynamic systems is impractical. Thus, we consider the real-time iteration scheme presented below. We show that under certain assumptions, the proposed iteration scheme can converge and even improve stability when the low-level MPC is solved using QP approximations, and the high-level problem is solved by taking gradient-based steps and a line-search.  

Given a fixed $\theta$, an initial condition, and an initial trajectory, we can create the following QP by linearizing the dynamics and constraints and taking a quadratic approximation of the cost about the initial trajectory:

\begin{align}
\label{eq:mpc_qp}
\min_{x, u}& \quad \frac{1}{2}\begin{bmatrix}
    x \\ u
\end{bmatrix}^TQ(\theta)\begin{bmatrix}
    x \\ u
\end{bmatrix} + q^T(\theta)\begin{bmatrix}
    x \\ u
\end{bmatrix} \notag \\
    \text{s.t.}& \quad  A(\theta)\begin{bmatrix}
        x \\ u
    \end{bmatrix} = b(\theta) \tag{A-MPC} \\
     & \quad G(\theta)\begin{bmatrix}
         x \\ u
     \end{bmatrix} \leq h(\theta) \notag
\end{align}
$A$ is the equality constraint matrix, $G$ is the inequality constraint matrix, and $b$ and $h$ are vectors defining the equality and inequality constraints respectively. $q$ and $Q$ are the gradient and the Hessian of the cost function respectively. We let $\Omega:=\{A,G,Q,b,h,q\}\subset\mathbb{R}^p$ and denote a specific choice of these parameters as $\omega \in \Omega$.

If $Q(\theta)$ is positive semidefinite, then \eqref{eq:mpc_qp} is a convex QP and can typically be solved quickly for real systems. Note \eqref{eq:mpc_qp} refers to the approximate MPC problem whereas \eqref{eq:ll_opt} refers to the true low-level MPC problem.

\begin{assumption}
\label{assumption:costs}
    The high-level cost satisfies $H(x, u, \theta) = L(x, u, \theta)$.
\end{assumption}

\begin{definition}
    Let $s \in \mathbb{R}^n$ be a state and $s^*$ be a fixed point of a discrete dynamical system. The fixed point is stable if
    \begin{equation}
    \label{eq:stability}
        \forall \epsilon > 0, \; \exists \delta \; \text{s.t.} \; \norm{s^{0} - s^*} < \delta \Rightarrow \norm{s^k - s^*} < \epsilon, \; \forall k > 0,
    \end{equation}
    and asymptotically stable if for some $\delta$,
    \begin{equation}
    \label{eq:asymp_stability}
        \norm{s^{0} - s^*} < \delta \Rightarrow \lim_{k \rightarrow \infty}\norm{s^k - s^*} = 0.
    \end{equation}
\end{definition}
Let $J^L(x_0^k, \theta)$ denote the minimum value of \eqref{eq:ll_opt} and $J^A(x_0^k, \theta)$ the minimum value of \eqref{eq:mpc_qp}, where $x_0^k$ is the initial condition at the $k$th time step.
\begin{assumption}
\label{assumption:mpc_stability}
There is a set of parameters, $\Theta \subseteq \mathbb{R}^c$, such that the low-level MPC controller \eqref{eq:mpc_qp} is stabilizing with a region of attraction $\mathcal{A} \subset \mathbb{R}^n$. Moreover, $J^L(x_0^k, \theta)$ is a discrete-time Lyapunov function certifying stability for all $\theta \in \Theta$ (in the sense of \eqref{eq:stability}) to the origin of the system. Thus, the following are satisfied:
\begin{align}
     \label{eq:lyap-1}
     J^L(x_0^k, \theta) &> 0, \quad \forall x_0^k \in \mathcal{A} \setminus 0, \\
     J^L(0, \theta) &= 0, \\
     \label{eq:layp-2}
     J^L(x_0^{k+1}, \theta) - &J^L(x_0^{k}, \theta) \leq 0. 
\end{align}
\end{assumption}

The high-level problem is solved throughout time by using the following update at each time step
\begin{align}
\label{eq:hl-step}
    \theta^{k+1} = \theta^k + \alpha^k p^k
\end{align}
\noindent where $\alpha^k$ gives the scaling of the step and $p^k$ is the step direction. The choice of step direction is based on gradient information from the low-level MPC problem \eqref{eq:mpc_qp}. To get this gradient information, we will differentiate through the solution of the QP.

\begin{proposition} 
\label{prop:qp_grad}
    Let the QP subproblem associated with $\omega(\theta) \in \Omega$ be denoted as 
    \begin{align}
    \label{eq:qp_subprob}
    \min_{z}& \quad \frac{1}{2}z^TQ(\theta)z + q^T(\theta)z \notag \\
        \text{s.t.}& \quad  A(\theta)z = b(\theta)  \\
         & \quad G(\theta)z \leq h(\theta) \notag
    \end{align}
    where $\theta$ is a parameter. Let $J$ denote the optimal cost. If at the optimal solution $(z^*, \lambda^*, \nu^*)$ (with Lagrange multipliers $\lambda^*$ and $\nu^*$) $\omega$ is smooth with respect to $\theta$; $Q \succ 0$; the Linear Independence Constraint Qualification (LICQ)\footnote{See definition 12.4 in \cite{nocedal_numerical_2006}.} condition holds; and the set $\{i \; | \; \lambda_i^* = 0 \; \text{and} \; (G(\theta)z^* - h(\theta))_i = 0\}$ is empty;
    then the partial derivative of $J$ with respect to a parameter element $\theta_j$ is given by 

    \begin{multline}
    \label{eq:cost_partial}
        \frac{\partial J}{\partial \theta_j} = (Q(\theta)z^* + q(\theta))^T\frac{\partial z^*}{\partial \theta_j} + \frac{1}{2}z^{*T}\frac{\partial Q}{\partial \theta}z^* + \frac{\partial q^T}{\partial \theta_j} z^*,
    \end{multline}
    where in particular,
    \begin{multline}
    \label{eq:imp_fcn_partial}
        \frac{\partial z^*}{\partial \theta_j} = -\begin{bmatrix}
            Q & G^T & A^T \\
            \diag(\lambda^*)G & \diag(Gz^* - h) & 0 \\
            A & 0 & 0
        \end{bmatrix}^{-1} \\
        \times \begin{bmatrix}
            \frac{\partial Q}{\partial \theta_j} z^* + \frac{\partial q}{\partial \theta_j} + \frac{\partial G^T}{\partial \theta_j} \lambda^* + \frac{\partial A^T}{\partial \theta_j}\nu \\
            \diag(\lambda^*)(\frac{\partial G}{\partial \theta_j}z^* - \frac{\partial h}{\partial \theta_j}) \\
            \frac{\partial A}{\partial \theta_j} z^* - \frac{\partial b}{\partial \theta_j}
        \end{bmatrix}
    \end{multline}
    with the dependence on $\theta$ suppressed for notational clarity.
\end{proposition}

\begin{proof}
    The proof follows from the chain rule and the results of \cite{barratt_differentiability_2019} and \cite{amos_optnet_2021}, which apply the implicit function theorem to the KKT conditions of \eqref{eq:qp_subprob}.
\end{proof}
We can form $\nabla_\theta J^{A}$ as a vector of partial derivatives computed via Prop. \ref{prop:qp_grad}. Note the matrix inversion in \eqref{eq:imp_fcn_partial} can be calculated once for a given $\theta$ and reused for each element of $\nabla_\theta J^{A}$.

If \eqref{eq:mpc_qp} coincides with the QP approximation of \eqref{eq:ll_opt} at the minimizer, then the gradients will match the true gradients of the MPC subproblem. However, we can only expect that \eqref{eq:mpc_qp} is a good approximation of the true MPC problem. Since the gradients of \eqref{eq:mpc_qp} are what we can compute in real-time, we show how to leverage them.

\subsection{Main Result}

Denote the program parameters of the QP approximation of the true MPC problem at the minimizer by $\omega^* \in \Omega$. Let $\innerp{\cdot , \cdot}$ represent the inner product.
\begin{lemma}
\label{lemma:orthog}
Assume $\norm{\nabla_{\theta} J^L(x_0^j, \theta^k)}$, the 2-norm of the optimal cost gradient of the true MPC, is nonzero. If $\nabla_{\theta}J^L(x_0^j, \theta^k)$ is continuous with respect to $\omega^* \in \Omega$, then there exists a ball of radius $\delta$ about 
$\omega^*$ such that if $\norm{\omega - \omega^*}^2 \leq \delta$, then $\innerp{\nabla_{\theta}J^{A}(x_0^j, \theta^k), \nabla_{\theta}J^L(x_0^j, \theta^k)} > 0$.
\end{lemma}

\begin{proof}
    We omit the dependence on $(x_0^j, \theta^k)$ in the proof for readability. By definition of the 2-norm we have that 
    \begin{equation*}
        \norm{\nabla_{\theta} J^{A} - \nabla_{\theta} J^L}^2 = \norm{\nabla_{\theta} J^{A}}^2 + \norm{\nabla_{\theta} J^L}^2 - 2\innerp{\nabla_{\theta} J^{A}, \nabla_{\theta} J^L}.
    \end{equation*}
    By continuity of the gradient $\nabla_{\theta} J^L$ with respect to $\omega^*$ we have that for all $\epsilon > 0$, there exists a $\delta > 0$ such that $\norm{\omega - \omega^*} < \delta \Rightarrow \norm{\nabla_{\theta} J^{A} - \nabla_{\theta} J^L} < \epsilon$. Thus, we can choose $\epsilon > 0$ such that
    \begin{equation*}
        (\norm{\nabla_{\theta} J^{A}}^2 + \norm{\nabla_{\theta} J^L}^2 - \epsilon)/2 < \innerp{\nabla_{\theta} J^{A}, \nabla_{\theta} J^L}.
    \end{equation*}
    Since $\norm{\nabla_{\theta} J^L}^2 > 0$ by assumption, there exists a $\epsilon > 0$ such that $0 < \innerp{\nabla_{\theta} J^{A}, \nabla_{\theta}J^L}$. Thus from the continuity of $\nabla_{\theta} J^L$ with respect to $\omega^*$, there exists a $\delta > 0$ such that $0 < \innerp{\nabla_{\theta} J^{A}, \nabla_{\theta}J^L}$ is satisfied.
\end{proof}

Lemma \ref{lemma:orthog} tells us that there is a region in $\Omega$ where if \eqref{eq:mpc_qp} is ``close'' to the QP approximation of \eqref{eq:ll_opt} at the minimizer, then the gradient we compute is bounded away from orthogonality with the true gradient. This will be useful for providing us with a descent direction of the cost function.

Given a current initial condition, we can substitute the MPC constraint into the high-level cost function to make the cost become solely a function of $\theta$. Since the initial condition changes at each time step, the MPC problem also changes, so the cost function changes with each iteration as a function of $\theta^k$.
We denote this cost with MPC embedded with the notation $H_j(\theta^k) := H(x_0^j, \theta^k)$. Similarly, we will use $J^{(\cdot)}_j(\theta^k) := J^{(\cdot)}(x_0^j, \theta^k)$.

The equations
\begin{align}
    \label{eq:wolfe-1}
    &H_k(\theta^k + \alpha^k p^k) \leq H_k(\theta^k) + c_1\alpha^k\nabla_{\theta}H_k(\theta^k)^Tp^k, \\
    \label{eq:wolfe-2}
    &\nabla_{\theta}H_k(\theta^k + \alpha^k p^k)^T p^k \geq c_2 \nabla_{\theta}H_k(\theta^k)^Tp^k,
\end{align}
are called the Wolfe Conditions and will be used to determine the line search magnitude. Note that $\alpha^k > 0 \; \forall k$, $0 < c_1 < c_2 < 1$. See \cite{nocedal_numerical_2006} for a proof that we can always satisfy the Wolfe Conditions under the conditions presented here.

Note that to evaluate $H_j(\theta^k)$, and thus to perform the line search, we require an MPC solve. The line search can be parallelized, as discussed in the next section, to minimize its computational cost. To get the true cost function value, one would need to solve MPC to convergence. Note that the proposed iterative update is still cheaper than solving the entire bilevel problem to convergence at each time step, and thus better for bilevel real-time MPC than current methods. As will be shown below, we do not need to update the parameter after every MPC solve, so we can still use real-time iterations of the MPC at every other time step. In practice, we may choose to only perform one or two \eqref{eq:mpc_qp} solves as an approximation of the cost.

\begin{assumption}
    \label{assumption:technical_convergence}
    $\nabla_{\theta}J^L_k$ is Lipschitz continuous in $\theta$ for all $k$.
\end{assumption}

\begin{lemma}
    \label{thm:convergence}
    Let the conditions in Prop. \ref{prop:qp_grad}, Lemma \ref{lemma:orthog}, and Assumptions 
    \ref{assumption:costs}, \ref{assumption:mpc_stability}, and \ref{assumption:technical_convergence} hold. Let $x^0_0 \in \mathcal{A}$, $\Theta = \mathbb{R}^c$, and $\mathcal{T} = \mathbb{R}^c$. Further, let $\alpha^k$ satisfy the Wolfe Conditions, and $p^k = -\nabla_{\theta}J_k^A(\theta^k)$ be calculated using Prop. \ref{prop:qp_grad}. If at each iteration $\omega^k$ is in the ball presented in Lemma \ref{lemma:orthog}, then $\norm{\nabla_{\theta}J^L_k(\theta^k)}^2 \rightarrow 0$ as $k \rightarrow \infty$.
\end{lemma}

Lemma \ref{thm:convergence} gives conditions under which the parameter optimization updates will converge to 0. This shows that under nominal conditions, when the system is stabilized, the parameter will converge to a critical point. Note that we can run multiple solves of \eqref{eq:mpc_qp} in between parameter optimizations and still converge, which will be made clear through the proof.

\begin{proof}
    From \eqref{eq:wolfe-2} and Assumption \ref{assumption:costs}, we have
    \begin{equation*}
        (\nabla_{\theta}J^L_k(\theta^{k+1}) - \nabla_{\theta}J^L_k(\theta^k))^Tp^k \geq (c_2 - 1)\nabla_{\theta}J^L_k(\theta^k)^Tp^k.
    \end{equation*}
    The Lipschitz condition gives us the following:
    \begin{equation*}
        (\nabla_{\theta}J^L_k(\theta^{k+1}) - \nabla_{\theta}J^L_k(\theta^k))^Tp^k \leq \alpha^k \Gamma \norm{p^k}^2,
    \end{equation*}
    where $\Gamma$ is the Lipschitz constant. Combining these two equations, we have
    \begin{equation*}
        \alpha_k \geq \frac{c_2 - 1}{\Gamma} \frac{(\nabla_{\theta}J^L_k(\theta^k)^Tp^k)^2}{\norm{p^k}^2}.
    \end{equation*}
    By \eqref{eq:wolfe-1}, we have
    \begin{equation}\label{eq:decrease_statement}
        J^L_{k}(\theta^{k+1}) \leq J^L_k(\theta^k) - c \cos^2(\mu^k) \norm{\nabla_{\theta}J^L_k(\theta^k)}^2,
    \end{equation}
    where $c = c_1(1 - c_2)/\Gamma$ and $\cos(\mu^k) = \frac{-\innerp{\nabla_\theta J^L_k(\theta^k), p^k}}{\norm{\nabla_{\theta}J_k^L} \norm{p^k}}$.
    
    Inequality \eqref{eq:decrease_statement} characterizes the decrease in cost for a single time step $k$. We further need to characterize the cost over time. By Assumption \ref{assumption:mpc_stability}, we know that $J^L$ is a Lyapunov function certifying stability of the system. Using \eqref{eq:layp-2}, we know that between time steps $J^L$ is non-increasing and thus $J^L_{k+1}(\theta^{k+1}) \leq J^L_{k}(\theta^{k+1})$. Note that we can perform any number of \eqref{eq:mpc_qp} solves between parameter updates here and arrive at the same inequality. Therefore,
    \begin{equation}
    \label{eq:cost_decrease}
        J^L_{k+1}(\theta^{k+1}) \leq J^L_k(\theta^k) - c \cos^2(\mu^k) \norm{\nabla_{\theta}J^L_k(\theta^k)}^2.
    \end{equation}
    Since \eqref{eq:cost_decrease} holds for all $k$, we get
    \begin{equation*}
        J^L_{k+1}(\theta^{k+1}) \leq J^L_0(\theta^0) - c \sum_{j = 0}^k \cos^2(\mu^j) \norm{\nabla_{\theta}J^L_j(\theta^j)}^2.
    \end{equation*}
    
    From the Lyapunov properties, $J^L$ is positive definite and thus bounded below, so $J^L_0(\theta_0) - J^L_{k+1}(\theta^{k+1}) < \eta$ for some $\eta \in \mathbb{R}$ and for all $k$. This implies that $\sum_{j = 0}^\infty \cos^2(\mu^j) \norm{\nabla_{\theta}J^L_j(\theta^j)}^2 < \infty$. 
    
    Now we can apply the result from Lemma \ref{lemma:orthog}. By assumption we satisfy $\norm{\omega_k - \omega_k^*} < \delta$ at each iteration and thus $\innerp{\nabla_{\theta}J^A_k, \nabla_{\theta}J^L_k} > 0$ for each $k$. From the definition of $p^k$, $\cos(\mu^k)$, and $\innerp{\nabla_{\theta}J^A_k, \nabla_{\theta}J^L_k} > 0$ it follows that $\cos^2(\mu^k) \geq \gamma > 0$ for all $k$. Combining this with $\sum_{j = 0}^\infty \cos^2(\mu^j) \norm{\nabla_{\theta}J^L_j(\theta^j)}^2 < \infty$ implies that $\norm{\nabla_{\theta}J^L_k(\theta^k)}^2 \rightarrow 0$ as $k \rightarrow \infty$.
\end{proof}

Leveraging Lemmas \ref{lemma:orthog} and \ref{thm:convergence}, we can determine how the high-level optimization affects the stability of the system. Define a connected set $\mathcal{B} \subseteq \mathcal{A}$ where $0 \in \mathcal{B}$.

\begin{theorem}
    \label{cor:stability}
    Let the conditions in Lemma \ref{thm:convergence} hold. If additionally, $\norm{\nabla_{\theta}J^L(x_0)} > 0 \; \forall x_0 \in \mathcal{B} \setminus {0}$, $\norm{\nabla_{\theta}J^L(x_0)} = 0$ at $x_0 = 0$, and the MPC subproblem is only stable in $\mathcal{B}$, then with the parameter update \eqref{eq:hl-step}, the system becomes asymptotically stable to the origin in the sense of \eqref{eq:asymp_stability} with region of attraction $\mathcal{B}$.
\end{theorem}

\begin{proof}
    By assumption, $J^L$ is a Lyapunov function that certifies stability. It is thus positive definite and $J^L_{k+1} - J^L_k \leq 0$. From Lemma \ref{thm:convergence}, and specifically \eqref{eq:cost_decrease}, the cost function decreases by at least $c \cos^2(\mu^k) \norm{\nabla_{\theta}J^L_k(\theta^k)}^2$ with each time step. By Lemma \ref{lemma:orthog} and the assumption on the gradient, we get that $c \cos^2(\mu^k) \norm{\nabla_{\theta}J^L_k(\theta^k)}^2 > 0$ for all $x_0 \in \mathcal{B} \setminus {0}$. Therefore, as an immediate result of the parameter optimization with line search satisfying the Wolfe Conditions, we get that $J^L_{k+1} - J^L_k < 0 \; \forall x_0 \in \mathcal{B} \setminus {0}$. This allows $J^L$ to now certify asymptotic stability of the origin.
\end{proof}

Theorem \ref{cor:stability} shows that there exist scenarios in which adding the additional parameter optimization \eqref{eq:hl_opt} and line search promotes the stability certificate from stable to asymptotically stable.

\subsection{Extension for Constraints}
We describe how to extend the convergence result above to constraints in the high-level problem enforced with log barriers. The convergence proof follows directly from the proof of Lemma \ref{thm:convergence} with two modifications. First, we must show that even without perfect gradient knowledge, we can still obtain a descent direction from differentiating the QP. Second, we must show that the Wolfe Conditions can always be satisfied given a descent direction $p^k$, because the cost function now has a bounded domain.

\begin{assumption}
\label{assumption:barrier_cost}
    The high-level cost satisfies $H(x, u, \theta) = L(x, u, \theta) + B(\theta)$ where $B(\theta)$ enforces a log barrier on the inequality constraints encoded in $\mathcal{T} \subset \mathbb{R}^c$. Specifically,
    \begin{equation*}
        B(\theta) = \sum_{j = 0}^{n_c} -\ln(g_j(\theta)),
    \end{equation*}
    where $g(\theta) \geq 0 \Leftrightarrow \theta \in \mathcal{T}$ and $n_c$ represents the number of inequality constraints (with the inequality evaluated element-wise). Finally, let $\mathcal{T}$ define a compact set.
\end{assumption}

\begin{lemma}
    \label{lemma:grad_accurate}
    Let Assumption \ref{assumption:barrier_cost} hold. Further, assume $\norm{\nabla_{\theta} J^L(x_0^j, \theta^k)} \neq 0$. If $\nabla_{\theta}J^L(x_0^j, \theta^k)$ is continuous with respect to $\omega^* \in \Omega$, then there exists a ball of radius $\delta$ about $\omega^*$ in $\Omega$ such that if $\norm{\omega - \omega^*}^2 \leq \delta$ then $\innerp{\nabla_{\theta}H^A(x_0^j, \theta^k), \nabla_{\theta}H(x_0^j, \theta^k)} > 0$, where $H^A(x_0^j, \theta^k) = J^A(x_0^j, \theta^k) + B(\theta^k)$.
\end{lemma}
\begin{proof}
    The gradient of the high-level cost function with a specific parameter and initial condition is $\nabla_\theta H_j(\theta^k) = \nabla_\theta J^L_j(\theta^k) + \nabla_\theta B(\theta^k)$. We have perfect knowledge of $\nabla_\theta B(\theta^k)$ and require $\innerp{\nabla_\theta J^A_j(\theta^k) + \nabla_\theta B(\theta^k), \nabla_\theta J^L_j(\theta^k) + \nabla_\theta B(\theta^k)} > 0$.
    
    We know if $\nabla_\theta J_j^A(\theta^k) = \nabla_\theta J_j^L(\theta^k)$, then the above condition is satisfied. Therefore, by continuity of the inner product and continuity of $\nabla_\theta J^L_j$ with respect to $\omega^*$, there exists a ball in $\Omega$ such that $\nabla_\theta H^A(x_0^k, \theta^k)$ satisfies the desired property.
\end{proof}

\begin{remark}
    Note that the conditions in Lemma \ref{lemma:grad_accurate} and Lemma \ref{lemma:orthog} are different. Satisfying only the condition in Lemma \ref{lemma:orthog} does not imply satisfaction of Lemma \ref{lemma:grad_accurate} and thus to satisfy this condition, we may require a smaller ball.
\end{remark}

\begin{lemma}
    \label{lemma:wolfe_sat}
    Let Assumption \ref{assumption:barrier_cost} hold. There exists an $\alpha$ such that the Wolfe Conditions can always be satisfied for $H$ with a descent direction $p$.
\end{lemma}
\begin{proof}
    Consider the function
    \begin{equation}
        \label{eq:wolfe_intersect}
        \phi(\alpha) = H_k(\theta + \alpha p) - H_k(\theta) - \alpha c_1 \nabla_\theta H_k(\theta)^T p. 
    \end{equation}
    Note that $\phi(0) = 0$. Since $c_1 < 1$ we get $\phi'(0) = (1 - c_1)\nabla_\theta H_k(\theta)^Tp < 0$ and therefore, $\exists \Tilde{\alpha} > 0$ such that $\phi(\tilde{\alpha}) < 0$. By the definition of a descent direction, $\alpha c_1 \nabla_\theta H_k(\theta)^T p < 0, \; \forall \alpha > 0$.

    Due to $\mathcal{T}$ being a compact set, there exists an $\alpha$ such that $\theta + \alpha p \notin \mathcal{T}$. Denote $\alpha'$ as the value such that $g(\theta + \alpha' p) = 0$. One can see that although $H_k(\theta) + \alpha c_1 \nabla_\theta H_k(\theta)^T p$ is bounded for all $\alpha \in [0, \alpha']$, we get that $\phi(\alpha) \rightarrow \infty$ as $\alpha \rightarrow \alpha'$ due to the log barrier approaching infinity, and thus $\exists \alpha^* \in (0, \alpha']$ such that $\phi(\alpha^*) = 0$. Then, as above, \eqref{eq:wolfe-1} is satisfied for all $\alpha \in (0, \alpha^*]$.

    Now, we can apply the mean value theorem as is done in classical proofs of Wolfe's Condition \cite{nocedal_numerical_2006} to get
    \begin{equation*}
        H_k(\theta + \alpha' p) - H_k(\theta) = \alpha' \nabla_\theta H_k(\theta + \alpha'' p)^Tp
    \end{equation*}
    for some $\alpha'' \in (0, \alpha^*)$. We can combine this with $\phi(\alpha^*) = 0$ to get
    \begin{equation*}
        \nabla_\theta H_k(\theta + \alpha'' p)^T p = c_1\nabla_\theta H_k(\theta)^T p > c_2\nabla_\theta H_k(\theta)^T p.
    \end{equation*}
    Thus, $\alpha''$ satisfies \eqref{eq:wolfe-1} and \eqref{eq:wolfe-2}.
\end{proof}
Here we present a convergence result for the constrained system.

\begin{proposition}
    \label{thm:constraints}
    Let Assumptions \ref{assumption:mpc_stability}, \ref{assumption:technical_convergence}, and \ref{assumption:barrier_cost} and the conditions in Lemma \ref{lemma:grad_accurate} and \ref{lemma:wolfe_sat} hold. Further, let $\mathcal{T} \subseteq \Theta$, $\theta^0 \in \mathcal{T}$, and $x_0^0 \in \mathcal{A}$. If $\omega^k$ is in the ball defined by Lemma \ref{lemma:grad_accurate} for all $k$, $p^k = -\nabla_\theta H_k^A(\theta^k)$ using the gradient from Prop. \ref{prop:qp_grad}, and $\alpha^k$ satisfies the Wolfe Conditions, then $\norm{\nabla_{\theta}H_k(\theta^k)}^2 \rightarrow 0$ as $k \rightarrow \infty$ and $\theta^k \in \mathcal{T}$ for all $k$.
\end{proposition}
\begin{proof}
Follows directly from the proof of Lemma \ref{thm:convergence}.
\end{proof}
Prop. \ref{thm:constraints} gives conditions under which we can expect the bilevel optimization to converge when the high-level optimization has constraints enforced via a log barrier. This also lets us relax the condition that $\Theta = \mathbb{R}^c$ if $\mathcal{T}$ is contained in the set.

\section{Application to Gait Generation}
In this section, we describe how to implement our proposed bilevel optimization algorithm for a legged robot in the context of contact schedule generation.

Fundamental to this approach is the observation that each foot is in a binary contact state. Thus, for each foot we can consider the times at which the contact is made and broken, i.e., when the contact state changes. This allows us to parameterize the entire set of contact schedules through these timings. We therefore let the parameter $\theta$ represent the contact state change times (i.e., lift off and touch down) of each foot throughout the horizon. This assumes we know what the robot will contact; thus, given a single surface, such as the ground, our method works well. 

\subsection{High-Level Optimization}
Associated with foot $i$ are the lift off and touch down times $\{\theta_{i,j}\}_{j=1}^{C}$, where $C$ is a user-defined parameter specifying how many contact changes to optimize. The order of the contact changes is enforced via the polytopic constraint set
\begin{equation*}
    \mathcal{T} = \{ \theta \; | \; t \leq \theta_{i,j} \leq \theta_{i,j+1} - k_{\text{min}}, \; \forall i, j < C\}.
\end{equation*}
Here, $t$ is the current time and $k_{\text{min}} > 0$ gives the minimum time for any contact state. We also enforce $\theta_{i,C} \leq k_{\text{end}} + \theta_{i,1}$, where $k_{\text{end}}>0$ is a user-defined parameter constraining the maximum time for the last contact relative to the first one. If a contact change time is in the past, i.e., $t > \theta_{i,j}$, then it is no longer optimized. For any leg not in contact, we only optimize contact times after the next touch down to prevent requesting instantaneous changes in the foot location. New contact times are added into the optimization as needed to fill the horizon.

Rather than solving \eqref{eq:imp_fcn_partial}, we can instead solve for gradients of the cost with respect to $\omega$ directly using the methods of \cite{amos_optnet_2021}. This provides an alternative and equivalent formulation of the gradient. Then, instead of \eqref{eq:cost_partial} we can use the following, where $\odot$ represents element-wise multiplication:
\begin{multline}
\label{eq:cost_partial_reform}
    \frac{\partial J_k^A}{\partial \theta_j} = \sum_{i,j} \bigg(\frac{\partial J_k^A}{\partial Q} \odot \frac{\partial Q}{\partial \theta_j}\bigg)_{(i,j)} +
    \sum_{i,j} \bigg(\frac{\partial J_k^A}{\partial A} \odot \frac{\partial A}{\partial \theta_j}\bigg)_{(i,j)} \\ +
    \sum_{i,j} \bigg(\frac{\partial J_k^A}{\partial G} \odot \frac{\partial G}{\partial \theta_j}\bigg)_{(i,j)} + \bigg\langle\frac{\partial J_k^A}{\partial q}, \frac{\partial q}{\partial \theta_j}\bigg\rangle \\ + \bigg\langle\frac{\partial J_k^A}{\partial b}, \frac{\partial b}{\partial \theta_j}\bigg\rangle + \bigg\langle\frac{\partial J_k^A}{\partial h}, \frac{\partial h}{\partial \theta_j}\bigg\rangle.
\end{multline}
To compute \eqref{eq:cost_partial_reform}, we still need to perform a large matrix inversion as given in \eqref{eq:imp_fcn_partial}, which can be performed via a sparse LU decomposition. 

Once the gradient information is computed, we can solve the Linear Program (LP)
\begin{align}
    \label{eq:lp}
    \min_{p^k}& \quad \nabla_{\theta}H_k^A(\theta)^T p^k \\
        \text{s.t.}& \quad p^k + \theta^k \in \mathcal{T} \notag
\end{align}
for the step direction $p^k$ of the contact times. Finally, we can perform a parallel line search to determine the step length modification, $\alpha^*$. Note that the line search requires an MPC solve for each value of $\alpha$; thus, to minimize the computation time, we parallelize the MPC solves. We solve the problem with contact schedules $\theta + \alpha_i p$ for a set of $\alpha_i$'s where $\theta + \alpha_i p \in \mathcal{T}$ and use the trajectory with the least cost. Note that this can be viewed as an approximation of a backtracking line search. For backtracking line search, we can in general ignore the Wolfe curvature conditions \eqref{eq:wolfe-2} \cite{nocedal_numerical_2006}, so this is an approximation of the line search mentioned above.

Alg. \ref{alg:bilevel_gait} gives the entire bilevel algorithm. $k_\text{hl}$ allows us to set the frequency of the gait optimization by only performing it once every $k_\text{hl}$ MPC solves. $k_\text{start}$ denotes the number of MPC solves for the starting trajectory.

\begin{algorithm}
    \caption{Bilevel Gait Generation}\label{alg:bilevel_gait}
    \begin{algorithmic}
    \State Given: $H(x, u)$, $\theta^0$, $k_{\text{end}}$, $k_{\text{min}}$, $x^0$, $u^0$, $k_\text{hl} \geq 1$, $k_\text{start} > 0$.
    \State Initialize $k = 0$ with $k_\text{start}$ iterations of MPC$(x^k, u^k, \theta^k)$
    \Loop
    \State Receive state estimates and modify $x^k$
    \If{$(k \mod k_\text{hl}) = 0$}
    \State Compute $\frac{\partial A}{\partial \theta}$, $\frac{\partial G}{\partial \theta}$, $\frac{\partial b}{\partial \theta}$, $\frac{\partial h}{\partial \theta}, \frac{\partial Q}{\partial \theta}, \frac{\partial q}{\partial \theta}$
    \State Form the gradient $\nabla_{\theta} J^A$ with \eqref{eq:cost_partial_reform}
    \State Solve \eqref{eq:lp} for the descent direction $p$
    \State Use parallel line search to compute $\alpha^*$
    \State $x^{k+1}, \; u^{k+1} = $\; QP-MPC$(x^k, u^k, \theta^k + \alpha^* p)$
    \State $\theta^{k+1} = \alpha^*p + \theta^k$
    \Else
    \State $x^{k+1}, \; u^{k+1} = $\; QP-MPC$(x^k, u^k, \theta^k)$
    \EndIf
    \State Send $x^{k+1}$ and $u^{k+1}$ to robot
    \State $k = k + 1$
    \EndLoop
    \Procedure{QP-MPC}{$x$, $u$, $\theta$}
    \State Form $A$, $b$, $G$, $h$ through linearization about $x$ and $u$
    \State Form $Q$, as a quadratic approximation of $L(x,u,\theta)$ and $q$ as its linear approximation about $x$ and $u$
    \State Solve the QP given in \eqref{eq:mpc_qp}
    \EndProcedure
    \end{algorithmic}
\end{algorithm}

\subsection{Model}
A nonlinear Single Rigid Body (SRB) model of the quadruped is used \cite{di_carlo_dynamic_2018}. This model reduces the quadruped to a single mass with fixed inertia. Since we only consider the main body of the quadruped, we make the approximation that the contact dynamics do not enter the model. This is generally a fair assumption if the feet contact the ground at low speeds. The dynamics are given as
\begin{equation}
    \begin{bmatrix} \dot{r} \\ \dot{l} \\ \dot{\xi} \\ I_R\dot{\zeta} \end{bmatrix} = \begin{bmatrix}
        l/m \\
        \sum_{i=1}^{n_e}F_i - mg \\
        \zeta \\
        -\zeta \times I_R\zeta + \sum_{i=1}^{n_e} (e_i - r)\times f_i
    \end{bmatrix}
    \label{eq:dynamics}
\end{equation}
where $r$ is the center of mass, $l$ the linear momentum, $\xi$ the orientation, and $I_R\zeta$ the angular velocity inertia matrix product. $m$ gives the mass of the robot, $g$ is the acceleration due to gravity, $n_e$ denotes the number of legs, $F_i$ is the force applied by the $i^{\text{th}}$ leg, and $e_i$ is the location of the $i^{\text{th}}$ leg in the world frame.

The orientation $\xi$ is represented as a quaternion. The states $\begin{bmatrix}r & l & \xi & I_R\zeta\end{bmatrix}^T$ are in $\mathbb{R}^{9} \times \mathbb{H}$ and we work with the derivative on the Lie algebra of the space, i.e. $\mathbb{R}^{12}$. Details can be found in \cite{csomay-shanklin_nonlinear_2023}.

\subsection{Spline Parameterization}
From Prop. \ref{prop:qp_grad}, we need $\omega$ to be smooth with respect to $\theta$. To ensure this smoothness, a spline parameterization will be used. Specifically, all variables that depend on the contact times are parameterized as splines consisting of Hermite polynomials connected end to end. For the model described above, these splines parameterize both the ground reaction forces and positions of the legs. Each individual leg's positions and forces share the same set of contact times.

Any segment of the spline can be expressed as
\begin{align}
    \sigma(t) &= a_0 + a_1 t + a_2 t^2 + a_3 t^3,
    \label{eq:hermite}
\end{align}
where the coefficients are determined by the spline's initial and final values ($y_0$, $y_1$), derivatives ($\dot{y}_0$, $\dot{y}_1$), and total time $T_s$, where $t \in [0, T_s]$. In the MPC problem, we optimize $(y_0, y_1, \dot{y}_0, \dot{y}_1)$, and in the high-level problem we optimize $T_s$ (which is given by $\theta$). This gives the coefficients
\begin{align*}
    a_0 &= y_0, \\
    a_1 &= \dot{y}_0, \\
    a_2 &= -T_s^{-2}(3(y_0 - y_1) + T_s(2\dot{y}_0 + \dot{y}_1)), \\
    a_3 &= T_s^{-3}(2(y_0 - y_1) + T_s(\dot{y}_0 + \dot{y}_1)).
\end{align*}
With this choice of parameterization, we can naturally encode constraints that would otherwise need to be explicitly stated in the QP. For example, the constraint enforcing no force at a distance is easily encoded by setting that polynomial segment to 0. For more details regarding this parameterization, see \cite{winkler_gait_2018}.

\subsection{Constraints}
Since the SRB model has no notion of the kinematics of the robot's legs but the MPC generates contact locations, we need a way to constrain the contact locations. Thus we add a box constraint for the leg location relative to the center of mass of the robot at each time node.

To prevent the MPC from requesting very large changes in position suddenly, we also add an equality constraint that enforces that the next touchdown location cannot be changed. This can be tuned to only be enforced after the leg has completed some percentage of the swing.

Lastly, we also enforce two ground reaction force constraints: a box constraint on the $z$ dimension of the force and a friction cone constraint. Let $F_b$ denote the maximum allowable ground reaction force and $F^x_i$, $F^y_i$, $F^z_i$ the $x$, $y$, $z$ forces for a given leg. These constraints are written
\begin{align*}
    0 \leq F_i^z &\leq F_b \quad 0 \leq i \leq n_e, \\
    | F_i^x | &\leq \mu F_i^z \quad 0 \leq i \leq n_e, \\
    | F_i^y | &\leq \mu F_i^z \quad 0 \leq i \leq n_e.
\end{align*}

\section{Simulation Results}
\begin{figure*}
\vspace{7pt}
    \centering
    \includegraphics[scale=0.28]{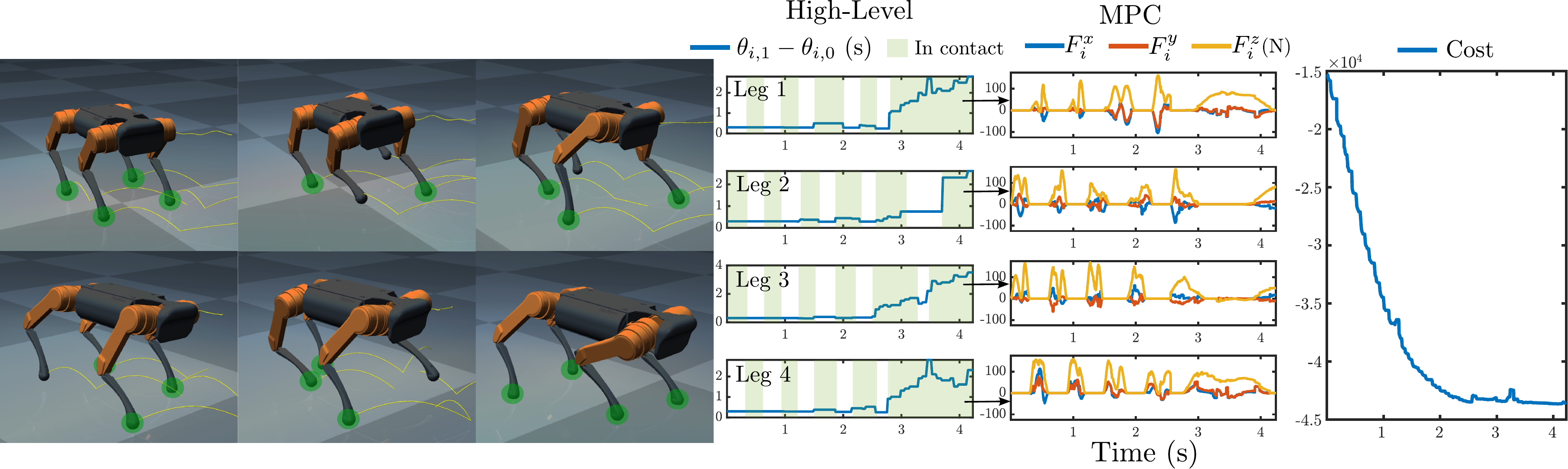}
\vspace{-6pt}
    \caption{Demonstration of diagonal walking. We observe that throughout the motion the robot achieves a number of different contact sequences. The green circles denote feet in contact with the ground. We start in a diagonal trot gait (diagonal legs move together) and by the end, the trajectory ``disappears" as the high-level optimizer adjusts the parameter such that its legs no longer move. The graph on the left shows the planned contact state times in blue and actual contact states in green. The middle graph plots MPC ground reaction forces which depend directly on the times computed in the high-level optimization on the left. The right-most graph shows the cost function over time.}
    \label{fig:gait-tile}
\end{figure*}

Alg. \ref{alg:bilevel_gait} was implemented in C++ and the code is publicly available online \footnote{\url{https://github.com/Zolkin1/bilevel-gait-gen}}. The QP in the MPC was solved with Clarabel \cite{goulart_clarabel_nodate} with solve tolerances of $1 \times 10^{-15}$, and the LP was solved with OSQP \cite{stellato_osqp_2020}. All derivatives were implemented analytically. The MPC cost was a quadratic function, but note that the full MPC problem is still a nonlinear program (due to the nonlinear constraints) and thus still must be approximated as a QP.

All results are simulated in MuJoCo \cite{todorov_mujoco_2012}. The simulation runs in real-time in a separate thread from the MPC to accurately reflect real-world behavior. A low-level QP controller that weights center of mass tracking, joint leg tracking (recovered via inverse kinematics), and ground reaction force tracking is used. This QP controller also enforces friction cone, maximum force, and holonomic constraints. The outputs of this controller are generalized accelerations and new desired ground reaction forces, both of which are then used to compute joint torques via inverse dynamics. This controller runs at 1 kHz.

The proposed MPC algorithm allows one to naturally account for early touchdowns with the feet. When a new touchdown is sensed we can update the internal contact schedule of MPC to note the contact.

\textbf{Computation Times}. Table \ref{tab:computation_times} gives the computation times for the implementation. The gradient column refers to the computation time required to create and factor the matrix in \eqref{eq:imp_fcn_partial}. The line search was computed using 10 parallel threads to check 10 evenly spaced values of $\alpha$ between 0 and 1. Unsurprisingly, using more nodes leads to a significant increase in computation time, although even at 50 nodes the algorithm has potential to run in real-time. We compare these times to other algorithms that can generate entire contact schedules on quadrupedal robots, such as \cite{kurtz_inverse_2023} at $60$ Hz (with 20 nodes), \cite{cleach_fast_2023} at $100$ Hz (with 3 nodes), and \cite{neunert_whole-body_2018} at 190 Hz (125 nodes). If we optimize the gait once every 5 MPC iterations (as done in Fig. \ref{fig:gait-tile}), then the average computation time with 20 nodes is about 90 Hz. One can see that the timings presented in this work are comparable to or faster than other algorithms.
\begin{table}
    \caption{Average computation times achieved for forward walking}
    \begin{center}
    \begin{tabular}{|c | c | c | c | c|} 
     \hline
     Nodes & $dt$ (s) & MPC (ms) & Gradient (ms) & Line Search (ms) \\  
     \hline
     20 & 0.05 & 8.0 & 3.8 & 19.5 \\ 
     33 & 0.033 & 13.9 & 7.0 & 35.9 \\
     50 & 0.02 & 25.8 & 14.1 & 72.8 \\
     \hline
    \end{tabular}
    \end{center}
    \label{tab:computation_times}
    \vspace{-0.5cm}
\end{table}

\textbf{Generation of New Contact Schedules}. Since the robot can generate its own contact schedules, when using a cost function incentivizing the robot to walk, we observe new gaits used throughout the motion to create a very natural result. Fig. \ref{fig:gait-tile} shows snapshots of the robot and data from both levels of the optimization throughout a diagonal walking motion.

\textbf{Disturbance Rejection}. We present a simulation study where the target behavior is standing still, but the robot is subject to an external force. Table \ref{tab:disturbance_rejection} shows the results of the trials. External forces were applied for 0.3 seconds. When applied in only one direction, the force was varied between 30 N to 50 N. When applied in both directions, the force was varied between 25 N to 45 N. The MPC without gait optimization only recovers from this disturbance in 9 of the 15 test cases. With the gait optimization, 13 of the 15 disturbances are successfully rejected, thus showing improvement over the generic MPC algorithm. Fig. \ref{fig:disturbance} shows snapshots of an example disturbance rejection.

\begin{table}
\caption{Disturbance Rejection Recovery}
\begin{center}
\begin{tabular}{| c | c | c |} 
    \hline
    Direction & High-Level Opt. & Without High-Level Opt.  \\  
    \hline
    $x$ & $5/5$ & $5/5$ \\
    $y$ & $5/5$ & $3/5$ \\
    $x$ and $y$ & $3/5$ & $1/5$ \\
    \hline
\end{tabular}
\end{center}
\label{tab:disturbance_rejection}
\vspace{-0.5cm}
\end{table}

\textbf{Optimality}. Table \ref{tab:costs} reports the decrease in the average cost of the motion relative to MPC without gait optimization. We can see that in all four of the test scenarios, an average cost decrease was achieved. This demonstrates the capability of new contact schedules to help lower desired costs.

Although the individual optimization problems are easier to solve than those present in many CIMPC approaches, this approach is still prone to landing in local minima. Thus, we observe that tuning the cost function can lead to significantly different gaits being generated. Similar to results in \cite{katayama_whole-body_2022}, we have also found that contact times are often (but not always) pushed out. Thus if one wants to avoid this behavior, one could add an additional constraint to the LP and/or modify the the cost in the gait optimization.

\begin{figure}
\vspace{6pt}
    \centering
    \includegraphics[scale=1]{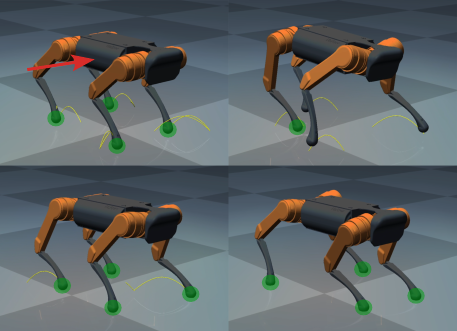}
    \vspace{-3pt}
    \caption{Snapshots of the robot while recovering from a disturbance in the $y$ direction, which is signified by the red arrow in the first tile. Observe in the upper right image that post-disturbance, three feet are off the ground while the original gait specified only two feet in the air at once. At the end, the robot's legs stop moving as one might naturally expect.}
    \label{fig:disturbance}
\end{figure}

\section{Conclusion}
In this paper, we have presented a real-time iteration scheme for a bilevel optimization with a low-level MPC that is suitable for use in real-world applications. Theoretical properties such as convergence and improvements in stability under nominal conditions were derived, justifying the approximations made by our proposed algorithm. Further, we proved that the use of our bilevel optimization algorithm can promote the stability of the system from stability in the sense of Lyapunov to asymptotic stability. We implemented the algorithm for contact schedule generation on a quadrupedal robot in simulation, demonstrating improved disturbance rejection, gait optimality, and generation of diverse gaits all while retaining real-time performance. Lastly, these results were achieved with computational speeds competitive with state of the art CIMPC algorithms. 

\begin{table}
    \caption{Cost Decrease Results}
    \begin{center}
    \begin{tabular}{| c | c |} 
     \hline
     Target Position & Avg. Cost Decrease  \\  
     \hline
     $(0.5, 0)$ & 3.2\% \\ 
     $(0.5, 0.5)$ & 2.32\% \\
     $(0, 0.5)$ & 5.6\% \\
     $(0, 0)$ & 6.44\% \\
     \hline
    \end{tabular}
    \end{center}
    \label{tab:costs}
    \vspace{-0.5cm}
\end{table}

\vspace{-1pt}
\bibliographystyle{IEEEtran}
\bibliography{BilevelGaitGeneration}

%\arxiv{\newpage
%\appendix
%\input{appendix}}{ }

\end{document}